\documentclass{amsart}
\usepackage{amsmath, amsfonts, amsthm, amssymb}
\usepackage{paralist}
\usepackage{graphics} %% add this and next lines if pictures should be in esp format
\usepackage{epsfig} %For pictures: screened artwork should be set up with an 85 or 100 line screen
\usepackage{graphicx}  
\usepackage{multicol}
\usepackage{epstopdf}

\usepackage{url}
\usepackage{pdfpages}
\usepackage[foot]{amsaddr}

\usepackage{tikz} 
\usetikzlibrary{shapes,arrows, patterns,shapes.geometric}

\newtheorem{theorem}{Theorem}[section]
\newtheorem{corollary}[theorem]{Corollary}
\newtheorem{lemma}[theorem]{Lemma}

\newtheorem{assumption}[theorem]{Assumption}

\newtheorem{proposition}[theorem]{Proposition}

\newtheorem{definition}[theorem]{Definition}

\newtheorem{example}[theorem]{Example}

\newtheorem{remark}[theorem]{Remark}

\title{Parity-check codes from disjunct matrices}

\author{Kathryn Haymaker$^1$}
\address{$^1$Department of Mathematics, Villanova University, 
Villanova, PA 19085}
\email{kathryn.haymaker@villanova.edu}

\author{Emily McMillon$^2$}
\address{$^2$Department of Mathematics, Rice University, Houston, TX 77005}
\email{em72@rice.edu}
\thanks{The second author was supported by the National Science Foundation grant DMS-2303380.}

\subjclass[2020]{Primary 94B05; % Linear codes (general)
Seconday 94B25, % Combinatorial codes
15B99, % Linear and multilinear algebra; matrix theory; special matrices; none of the above, but in this section
05B20. % combinatorial aspects of matrices (incidence, Hadamard, etc.)
}


\date{July 28, 2024}

\dedicatory{}

\begin{document}

\begin{abstract}
The matrix representations of linear codes have been well-studied for use as disjunct matrices. However, no connection has previously been made between the properties of disjunct matrices and the parity-check codes obtained from them. This paper makes this connection for the first time. We provide some fundamental results on parity-check codes from general disjunct matrices (in particular, a minimum distance bound). We then consider three specific constructions of disjunct matrices and provide parameters of their corresponding parity-check codes including rate, distance, girth, and density. We show that, by choosing the correct parameters, the codes we construct have the best possible error-correction performance after one round of bit-flipping decoding with regard to a modified version of Gallager's bit-flipping decoding algorithm.
\end{abstract}

\maketitle

\section{Preliminaries} \label{sec:prelims}

We will begin with some general notation we will use, then will move into sections on disjunct matrices and coding theory. Throughout, $[N]$ denotes the set $\{1, 2, \dotsc, N\}$. The finite field with $q$ elements is denoted $\mathbb{F}_q$. Vectors are denoted in boldface, i.e. $\mathbf{x}$. The Boolean sum of two binary vectors $\mathbf{x}$ and $\mathbf{y}$ is denoted $\mathbf{x} \vee \mathbf{y}$ and calculated componentwise, where $0 \vee 0 = 0$ and $0 \vee 1 = 1 \vee 0 = 1 = 1 \vee 1$. The support of a binary string $\mathbf{x}$ is the set of indices $i$ where $x_i=1$.  Finally, define $\mathbf{e}_i$ to be the $i$th standard basis (column) vector, where the length of the vector will be clear from context. 

\subsection{Disjunct Matrices} \label{sec:disjunct}

We will give some introductory definitions and background on disjunct matrices as well as separable matrices, a related class of matrices. Further results as well as some history can be found in \cite{du2000combinatorial}.

$D$-disjunct matrices are usually studied in the context of \textit{nonadaptive group testing}, in which a collection of samples, $\mathcal{S}$---each of which can be either positive or negative---is assigned to a group (a subset of $\mathcal{S}$) via a binary matrix $M$.  The columns of $M$ correspond to samples and the rows to groups; there is a one in position $M_{i,j}$ if and only if sample $j$ is contained in group $i$.  In group testing, the outcome vector for the groups records the positive/negative status of each group, and a decoding procedure is used to assign a positive or negative value to each sample based on the group outcomes. For the purposes of this work, suppose that a group is positive if it contains at least one positive sample, and negative otherwise. The definition of a $D$-disjunct matrix ensures a unique footprint among the samples so that negative samples can be detected even in the presence of up to $D$ other positives. Indeed, suppose that $D$ samples are positive and another distinguished one is negative. The negative sample can be identified as such if that sample is contained in some group that has none of the $D$ positive samples. 

This leads to the following definition.

\begin{definition} \label{def:disjunct}
    A matrix $M$ is \emph{$D$-disjunct} if, for any $D+1$ columns of $M$ with one column designated, there is a row with a $1$ in the designated column and a $0$ in each of the other $D$ columns.
\end{definition}

\begin{example}\label{ex:disjunct}
    The matrix $M$ below is $2$-disjunct. This means that, given any collection of $3$ columns, for any column in that set of $3$, there exists a row that has a $1$ while the other two columns have a $0$.
    \[ M = \begin{bmatrix} 1 & 1 & 1 & 0 & 0 & 0 & 0 & 0 & 0 & 0 \\ 
                           1 & 0 & 0 & 1 & 1 & 0 & 0 & 0 & 0 & 0 \\ 
                           0 & 1 & 0 & 1 & 0 & 1 & 0 & 0 & 0 & 0 \\
                           0 & 0 & 1 & 0 & 1 & 1 & 0 & 0 & 0 & 0 \\
                           1 & 0 & 0 & 0 & 0 & 0 & 1 & 1 & 0 & 0 \\
                           0 & 1 & 0 & 0 & 0 & 0 & 1 & 0 & 1 & 0 \\
                           0 & 0 & 1 & 0 & 0 & 0 & 0 & 1 & 1 & 0 \\
                           0 & 0 & 0 & 1 & 0 & 0 & 1 & 0 & 0 & 1 \\
                           0 & 0 & 0 & 0 & 1 & 0 & 0 & 1 & 0 & 1 \\
                           0 & 0 & 0 & 0 & 0 & 1 & 0 & 0 & 1 & 1 \end{bmatrix}.\]
    For example, consider the collection of columns $c_1$, $c_3$, and $c_6$ of $M$:
    \begin{center}
    \begin{tabular}{c c c}
    $c_1$ & $c_3$ & $c_6$ \\
    \hline
    $1$ & $1$ & $0$ \\
    $1$ & $0$ & $0$ \\
    $0$ & $0$ & $1$ \\
    $0$ & $1$ & $1$ \\
    $1$ & $0$ & $0$ \\
    $0$ & $0$ & $0$ \\
    $0$ & $1$ & $0$ \\
    $0$ & $0$ & $0$ \\
    $0$ & $0$ & $0$ \\
    $0$ & $0$ & $1$
    \end{tabular}
    \end{center}
    Column $c_1$ has a $1$ in row $2$ while $c_3$ and $c_6$ have a $0$. Column $c_3$ has a $1$ in row $7$ while $c_1$ and $c_6$ have a $0$, and column $c_6$ has a $1$ in row $3$ while $c_1$ and $c_3$ have a $0$. However, notice that if we added column $c_8 = (0,0,0,0,1,0,1,0,1,0)^T$ to this collection, there no longer exists a row of $c_3$ that has a $1$ all of the columns $c_1$, $c_6$, and $c_8$ have a $0$. This shows that $M$ is not $3$-disjunct.
\end{example}

Kautz and Singleton \cite{kautz1964nonrandom} proved that for any binary matrix, the  following holds (with credit to Guruswami, Rudra, and Sudan in \cite{essential20} for the updated notation): 

\begin{lemma}[\cite{kautz1964nonrandom}]
\label{lemma-disjunct} 
Let $1\leq D< N$  
be an integer and $M$ be a binary $t\times N$ matrix, such that 
\begin{enumerate} 
\item For every $j\in \{1, 2, \ldots, N\}$, the $j$th column has at least $w_{\min}$ ones in it. 

\item For every $i\neq j$ in $\{1, 2, \ldots, N\}$, the $i$ and $j$th columns have at most $a_{\max}$ ones in common 
\end{enumerate} 

for some integers $a_{\max}\leq w_{\min}\leq t$. Then $M$ is a 
\[ D=\left\lfloor \frac{w_{\min}-1}{a_{\max}} \right\rfloor\] disjunct matrix.  
\end{lemma} 

Of interest in the construction of parity-check codes is the density of the parity-check matrices. From \cite{dyachkov1982bounds}, given an arbitrary $t \times n$ $D$-disjunct matrix, $M$, $t < D (1 + o(1)) \ln (n)$. If we assume that $M$ has constant column weight $w$ (a common condition of algebraically constructed disjunct matrices), then $M$ has row weight $\frac{n w}{t} < \frac{nw}{D \ln(n)}$. In particular, the row weight of $M$ is at most $O\left( \frac{n}{\ln(n)}\right)$.

A class of matrices similar to disjunct matrices are \textit{separable matrices}. In these, pairs of sets of columns of size $D$ are compared to each other. In the context of group testing, the $D$-separable property implies that each collection of $D$ positive samples induces a unique footprint of groups that are positive.  By matching the set of positive groups with the appropriate samples, the $D$ positives can be identified. Both $D$-separable and $D$-disjunct matrices allow for $D$ positives to be identified, but the additional structure of $D$-disjunct matrices results in a reduction in decoding complexity. 

\begin{definition} \label{def:separable}
    A matrix $M$ is \emph{$D$-separable} (or \emph{$\bar{D}$-separable}) if no two sets of $D$ columns (resp.\ up to $D$ columns) have the same Boolean sum.
\end{definition}

Clearly, any matrix that is $\bar{D}$-separable is also $D$-separable. We conclude this subsection with a well-known result on the relationship between $D$-disjunct and $\bar{D}$-separable matrices.

\begin{theorem} [\cite{kautz1964nonrandom}]
    Any $D$-disjunct matrix is $\bar{D}$-separable and any $\bar{D}$-separable matrix is $(D-1)$-disjunct.
\end{theorem}

Further, a binary matrix is $\overline{D+1}$-separable if and only if it is both $(D+1)$-separable and $D$-disjunct \cite{hwang2006pooling}. 

\subsection{Coding Theory} \label{sec:codingtheory}

We will give some introductory definitions and background on linear codes. A good reference for additional information is \cite{huffman2010fundamentals}.

An $[n,k]_q$ \textit{linear error-correcting code} is a $k$-dimensional subspace of the $n$-dimensional vector space $\mathbb{F}_q^n$. This subspace can be defined as the kernel of an $m \times n$ matrix $H$ with row rank $n-k$, called a \textit{parity-check matrix} of the code. In other words, given a parity-check matrix $H$, the code $\mathcal{C}(H)$, or just $\mathcal{C}$ when $H$ is understood, is given by
\[ \mathcal{C}(H) = \{ \mathbf{x} \in \mathbb{F}_q^n \mid H \mathbf{x}^T = \mathbf{0}\}. \]

The \textit{minimum distance} of a linear code $\mathcal{C}$ is $d_{\text{min}}(\mathcal{C}) = \min \limits_{\mathbf{c} \in \mathcal{C}} \text{wt}_H(\mathbf{c})$ where $\text{wt}_H(\mathbf{c})$ is the \textit{Hamming distance}, or number of nonzero coordinates, of $\mathbf{c}$. If the minimum distance of a linear code $\mathcal{C}$ is known, it is called an $[n,k,d]_q$ code. Linear codes over $\mathbb{F}_2$ are called \textit{binary linear codes}. A binary linear code $\mathcal{C}$ is called an $[n,k,d]$ code (with the subscript $2$ omitted). A well-known and very useful fact is that the minimum distance of a code $\mathcal{C}(H)$ is the smallest size of a collection of linearly dependent columns of $H$.

The \textit{rate} of an $[n,k,d]_q$ code $\mathcal{C}$ is the ratio of information symbols to transmitted symbols, $\frac{k}{n}$. In this work, we are primarily concerned with codes defined by their parity-check matrices. If $H$ is an $m \times n$ parity check matrix for $\mathcal{C}$, note that its rate is given by $1 - \frac{n-k}{n}$. If $H$ has redundant rows (i.e. if $m > n-k$), then $1 - \frac{m}{n}$ gives a lower bound on the rate of $\mathcal{C}$. It is often not straightforward to determine the rank of a parity-check matrix, so results in this paper will focus on finding a lower bound on code rates.

Parity-check codes are often associated with a type of bipartite graph called a Tanner graph \cite{tanner1981recursive}. Given a parity-check matrix $H$, the bipartite graph $G = (V,W;E)$ such that the vertex set $V$ (the set of ``variable nodes'') corresponds to the columns of $H$, the vertex set $W$ (the set of ``check nodes'') corresponds to the rows of $H$, and $v_t \in V$ and $w_s \in W$ are adjacent if and only if the $(s,t)$ entry of $H$ is nonzero is called the \textit{Tanner graph} of $H$. In other words, $H$ is the biadjacency matrix of the graph $G$. An example of a parity-check matrix and its Tanner graph is given in Example~\ref{ex:hammingcode}.

We note that a single code always has multiple parity-check matrices, and so also always has multiple Tanner graph representations. The choice of parity-check matrix has an effect on iterative decoder performance. One important consideration in choice of parity-check matrix is the girth of the Tanner graph. The \textit{girth} of a graph $G$ is the length of a shortest cycle in $G$. For bipartite graphs the girth is even. Because iterative decoding is optimal on cycle-free Tanner graphs \cite{gallager1962low, wiberg1996codes}, large girths are desirable. At the very least, it is desirable for parity-check codes to have Tanner graphs free of $4$-cycles, so with girth at least $6$. In addition, over the binary erasure channel (BEC), iterative decoder failure is characterized by a class of Tanner graph substructures called \textit{stopping sets} \cite{di2002finite}.

\begin{definition}
    A \emph{stopping set} is a subset $S \subseteq V$ of variable nodes in a code's Tanner graph $G = (V,W)$ such that each neighbor of $S$ is connected to $S$ at least twice. Equivalently, a stopping set is a subset $\mathcal{V}$ of columns of a code's parity-check matrix such that each row, when restricted to $\mathcal{V}$, does not have weight $1$.  
\end{definition}

Of particular interest is the minimum size of a stopping set in a given code, which is dependent on the chosen parity-check matrix or Tanner graph representation.

\begin{definition}
    Given a parity-check matrix $H$, the the \emph{minimum stopping set size} or \emph{minimum stopping distance} of $\mathcal{C}(H)$ is denoted $s_{\text{min}}(H)$ and is the smallest size of any stopping set in $\mathcal{C}(H)$.
\end{definition}

It is clear that any codeword is also a stopping set. Consequently, for any parity check matrix $H$ of a code $\mathcal{C}$, $s_{\text{min}}(H) \leq d_{\text{min}}(\mathcal{C})$.  Because small stopping sets are probabilistically more common and impact iterative decoding, it is desirable that the smallest size of a stopping set in a given code be as large as possible. 
Hence, code representations for which $s_{\text{min}}(H) = d_{\text{min}}(\mathcal{C})$ are desirable.

\begin{example}\label{ex:hammingcode}
    The matrix $H$ below is the parity-check matrix for a $[7,4,3]$ binary code.
    \[ H = \begin{bmatrix} 0 & 0 & 0 & 1 & 1 & 1 & 1 \\
                           0 & 1 & 1 & 0 & 0 & 1 & 1 \\
                           1 & 0 & 1 & 0 & 1 & 0 & 1 \end{bmatrix}\]
    Its Tanner graph is given in Figure~\ref{fig:tannergraph}. The set of variable nodes $\mathcal{V} = \{v_1, v_4, v_5\}$ is a stopping set, because the set of neighbors of $\mathcal{V}$, $\{c_1, c_3\}$, is such that each of $c_1$ and $c_3$ is connected to at least two variable nodes in $\mathcal{V}$.
\end{example}

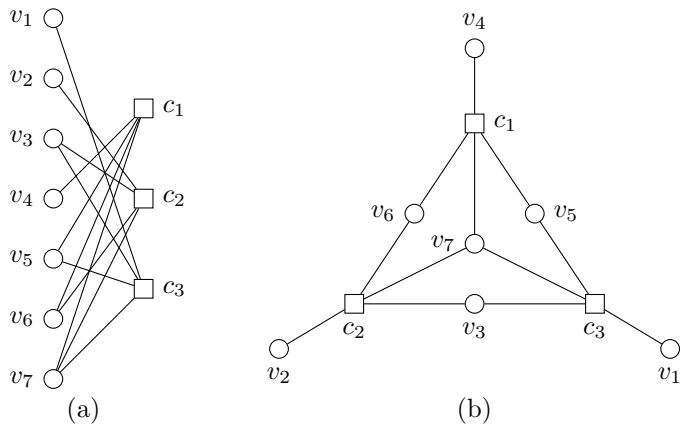
\begin{figure}[htb] 
\centering
\begin{tikzpicture}[scale=0.4,square/.style={regular polygon,regular polygon sides=4}]

    \node[fill=white,draw=black,circle,label=left:{$v_1$},scale=0.75] (v1) at (0,0) {};
    \node[fill=white,draw=black,circle,label=left:{$v_2$},scale=0.75] (v2) at (0,-2) {};
    \node[fill=white,draw=black,circle,label=left:{$v_3$},scale=0.75] (v3) at (0,-4) {};
    \node[fill=white,draw=black,circle,label=left:{$v_4$},scale=0.75] (v4) at (0,-6) {};
    \node[fill=white,draw=black,circle,label=left:{$v_5$},scale=0.75] (v5) at (0,-8) {};
    \node[fill=white,draw=black,circle,label=left:{$v_6$},scale=0.75] (v6) at (0,-10) {};
    \node[fill=white,draw=black,circle,label=left:{$v_7$},scale=0.75] (v7) at (0,-12) {};

    \node[fill=white,draw=black,square,label=right:{$c_1$},scale=0.75] (c1) at (3,-3) {};
    \node[fill=white,draw=black,square,label=right:{$c_2$},scale=0.75] (c2) at (3,-6) {};
    \node[fill=white,draw=black,square,label=right:{$c_3$},scale=0.75] (c3) at (3,-9) {};

    \draw[thin] (v1) to (c3);
    \draw[thin] (v2) to (c2);
    \draw[thin] (v3) to (c2);
    \draw[thin] (v3) to (c3);
    \draw[thin] (v4) to (c1);
    \draw[thin] (v5) to (c1);
    \draw[thin] (v5) to (c3);
    \draw[thin] (v6) to (c1);
    \draw[thin] (v6) to (c2);
    \draw[thin] (v7) to (c1);
    \draw[thin] (v7) to (c2);
    \draw[thin] (v7) to (c3);

    \node[label=below:{(a)}] at (1,-12) {};

    \node[fill=white,draw=black,circle,label=below:{$v_1$},scale=0.75] (v1) at (20.5,-11) {};
    \node[fill=white,draw=black,circle,label=below:{$v_2$},scale=0.75] (v2) at (7.5,-11) {};
    \node[fill=white,draw=black,circle,label=below:{$v_3$},scale=0.75] (v3) at (14,-9.5) {};
    \node[fill=white,draw=black,circle,label=above:{$v_4$},scale=0.75] (v4) at (14,-1) {};
    \node[fill=white,draw=black,circle,label=right:{$v_5$},scale=0.75] (v5) at (16,-6.5) {};
    \node[fill=white,draw=black,circle,label=left:{$v_6$},scale=0.75] (v6) at (12,-6.5) {};
    \node[fill=white,draw=black,circle,label=left:{$v_7$},scale=0.75] (v7) at (14,-7.5) {};

    \node[fill=white,draw=black,square,label=right:{$c_1$},scale=0.75] (c1) at (14,-3.5) {};
    \node[fill=white,draw=black,square,label=below:{$c_2$},scale=0.75] (c2) at (10,-9.5) {};
    \node[fill=white,draw=black,square,label=below:{$c_3$},scale=0.75] (c3) at (18,-9.5) {};

    \draw[thin] (v1) to (c3);
    \draw[thin] (v2) to (c2);
    \draw[thin] (v3) to (c2);
    \draw[thin] (v3) to (c3);
    \draw[thin] (v4) to (c1);
    \draw[thin] (v5) to (c1);
    \draw[thin] (v5) to (c3);
    \draw[thin] (v6) to (c1);
    \draw[thin] (v6) to (c2);
    \draw[thin] (v7) to (c1);
    \draw[thin] (v7) to (c2);
    \draw[thin] (v7) to (c3);

    \node[label=below:{(b)}] at (14,-12) {};

\end{tikzpicture}

\caption{(a) The Tanner graph of $H$ given in Example~\ref{ex:hammingcode}. (b) A planar representation of the graph in (a).} \label{fig:tannergraph}
\end{figure}

Two commonly studied classes of parity-check codes are low-density parity-check (LDPC) codes and moderate-density parity-check (MDPC) codes. These classes of parity-check codes are differentiated by the row weights of their parity-check matrices. A family of binary linear parity-check codes $\{\mathcal{C}_i\}$ with parity check matrices $H_i$ and lengths $n_i$ is called MDPC if the row weight of $H_i$ is $O(\sqrt{n_i})$ \cite{ouzan2009moderate}. LDPC codes are similarly defined, but there is less consensus on the requirements for the row weights of a family of LDPC codes. Definitions range from ``sparse relative to the length of the code'' to ``anything less than MDPC'' to ``row weights less than $10$''. For the purposes of this paper, we use the row weight less than $10$ definition (as in \cite{misoczki2013mdpc}). Finally, if the weight of every column of $H_i$ is some constant $v_i$ and the weight of any row of $H_i$ is some constant $w_i$, we say that $\mathcal{C}_i$ (whether LDPC or MDPC) is of type $(v_i,w_i)$.

The performance of a parity-check code under iterative decoding is highly dependent on its chosen parity-check matrix. An important property of a matrix to be used in bit-flipping decoding is the \textit{maximum column intersection}, which is defined as, for a binary matrix $H$, the maximum size of the intersection of the supports for any pair of columns of $H$. Recall from Section~\ref{sec:disjunct} that we denote this value $a_{\max}$. In the context of the Tanner graph of a given parity-check matrix, if $a_{\max} \geq 2$, the Tanner graph has $4$-cycles and so has girth $4$ (the smallest possible girth for a bipartite graph). If $a_{\max} = 1$, the Tanner graph has girth at least $6$.

The following result from Tillich \cite{tillich2018decoding} gives the amount of errors that can always be corrected within one round of a bit-flipping decoding algorithm under majority-logic decoding.

\begin{proposition}[\cite{tillich2018decoding}] \label{prop:tillich}
    Let $C$ be a code with parity check matrix $H$ such that every column has weight at least $v$ and whose maximum column intersection is $a_{\max}$. Performing majority-logic decoding based on this matrix corrects all errors of weight $\leq \left\lfloor \frac{v}{2 a_{\max}} \right\rfloor$.
\end{proposition}

In the context of parity-check matrices, those that minimize $a_{\max}$ while maximizing the column weights of a matrix yield optimal decoder performance. In Section~\ref{sec:specific}, we provide three methods of constructing parity-check matrices that minimize $a_{\max}$ (i.e. have $a_{\max} = 1$), and thus for fixed column weight $v$, have optimal decoder performance. We also provide the additional constraints required on these constructions to maximize their column weights (while still being MDPC). 

\section{General Results} \label{sec:generalresults}

In this section, we examine the use of general disjunct and separable matrices as parity-check matrices for codes. We begin by deriving the following simple lower bound on the minimum distance of such a code.

\begin{proposition} \label{prop:mindis}
    Let $M$ be an $m \times n$ $D$-disjunct matrix. Then the code $\mathcal{C}(M)$ has minimum distance $d_{\text{min}}(\mathcal{C}) \geq D+2$.
\end{proposition}

\begin{proof}
    Let $J \subseteq [n]$ be such that $|J| \leq D+1$. Because $M$ is $D$-disjunct, for any $i \in J$, there exists some row in the column vector $\mathbf{m}_i$ that is a $1$ while all column vectors in the collection $\{\mathbf{m}_j\}_{j \in J\setminus\{i\}}$ are equal to zero in that row. This implies that $\sum_{j \in J} \mathbf{m}_j \neq \mathbf{0}$. So any set of $D+1$ or fewer columns is not linearly dependent, and hence $d_{\text{min}}(\mathcal{C}) \geq D+2$.
\end{proof}

Note that if a matrix $M$ is $D$-disjunct but not $(D+1)$-disjunct, it does not necessarily have minimum distance $D+2$, and hence this is only a lower bound. See Example~\ref{ex:distance}.

\begin{example}\label{ex:distance}
    Consider the matrix
    \[ M = \begin{bmatrix} 1 & 0 & 0 & 1 \\ 0 & 1 & 0 & 1 \\ 0 & 0 & 1 & 1 \end{bmatrix}.\]
    $M$ is $0$-disjunct, but the code with parity check matrix $M$ has minimum distance $4$.
\end{example}

\begin{proposition}\label{prop:stoppingdis}
    Let $M$ be an $m \times n$ $D$-disjunct matrix. Then the code $\mathcal{C}(M)$ with parity check matrix $M$ has stopping distance $s_{\text{min}}(M) \geq D+2$.
\end{proposition}

\begin{proof}
    The proof is very similar to that of Proposition~\ref{prop:mindis}. Let $J \subseteq [n]$ be such that $|J| \leq D+1$. Because $M$ is $D$-disjunct, for any $i \in J$, there exists some row in the column vector $\mathbf{m}_i$ that is a $1$ while all other column vectors in the collection $\{\mathbf{m}_j\}_{j \in J\setminus\{i\}}$ are equal to zero in that row. In particular, this implies that for any collection of $D+1$ or fewer columns, some neighbor of the variable nodes indexed by $J$ is connected to the set of variable nodes only once. Hence $\mathcal{C}$ has no stopping sets of size $D+1$ or smaller, and $s_{\text{min}}(\mathcal{C}) \geq D+2$.
\end{proof}

One could use Proposition~\ref{prop:stoppingdis} to infer Proposition~\ref{prop:mindis}, because $s_{\text{min}} \geq d_{\text{min}}$.  However, the direct proof of Proposition~\ref{prop:stoppingdis} is still worth writing, as it is a fundamental result worthy of its own statement. In contrast, for $\bar{D}$-separable matrices, we begin by providing a lower bound on the stopping distance.

\begin{proposition} \label{prop:stoppingdis2}
    Let $M$ be an $m \times n$ $\bar{D}$-separable matrix. Then the code $\mathcal{C}(M)$ with parity check matrix $M$ has stopping distance $s_{min}(M) \geq D+2$.
\end{proposition}

\begin{proof}
    Because $M$ is $\bar{D}$-separable, it is $(D-1)$-disjunct, so by Proposition~\ref{prop:stoppingdis}, $s_{\text{min}}(M) \geq D+1$. We will show that $s_{\text{min}}(M) \neq D+1$.

    Suppose by way of contradiction that there exists some collection of $D+1$ columns of $M$, $\mathcal{M} = \{\mathbf{m}_i\}_{i=1}^{D+1}$, such that any row that has a $1$ in some column vector in $\mathcal{M}$ has at least two column vectors with a $1$ in that row. (In other words, a set of columns corresponding to a stopping set of weight $D+1$.) Because $M$ is $\bar{D}$-separable, the Boolean sum of any collection of $D$ of the columns in $\mathcal{M}$ is unique. 
    
    In a slight abuse of notation, let $\mathbf{m}_i = \{m_{i,1}, \dotsc, m_{i,m}\}^T$. Let $I \subseteq [m]$ be the set of indices such that $m_{1,i} = m_{2,i} = 1$.

    Consider the Boolean sum $\bigvee_{i=3}^{D+1} \mathbf{m}_i$. We claim that $\mathbf{m}_1 \lor \left( \bigvee_{i=3}^{D+1} \mathbf{m}_i \right) = \mathbf{m}_2 \lor \left( \bigvee_{i=3}^{D+1} \mathbf{m}_i \right)$.
    
    \begin{itemize}
        \item Consider first any row index $i \in I$. Because $m_{1,i} = m_{2,i} = 1$, it is clear that $m_{1,i} \vee \left( \bigvee_{j=3}^{D+1} m_{j,i}\right) = m_{2,i} \vee \left( \bigvee_{j=3}^{D+1} m_{j,i} \right) = 1$. 
        \item Now consider any row index $i \in [m] \setminus I$. There are two cases. If $m_{1,i} = m_{2,i} = 0$, it is clear that $m_{1,i} \vee \left( \bigvee_{j=3}^{D+1} m_{j,i}\right) = m_{2,i} \vee \left( \bigvee_{j=3}^{D+1} m_{j,i} \right) = \bigvee_{j=3}^{D+1} m_{j,i}$. If one of $m_{1,i}$ and $m_{2,i}$ is $1$ and one is $0$, because each row in $\mathcal{M}$ that has a $1$ has more than one $1$, there must be some $\mathbf{m}_k \in \{\mathbf{m}_{i}\}_{i=3}^{D+1}$ such that $m_{k,i} = 1$. Hence $m_{1,i} \vee \left( \bigvee_{j=3}^{D+1} m_{j,i}\right) = m_{2,i} \vee \left( \bigvee_{j=3}^{D+1} m_{j,i} \right) = 1$.
    \end{itemize}

    This shows that $\mathbf{m}_1 \lor \left( \bigvee_{i=3}^{D+1} \mathbf{m}_i \right) = \mathbf{m}_2 \lor \left( \bigvee_{i=3}^{D+1} \mathbf{m}_i \right)$, contradicting that $M$ is $\bar{D}$-separable. So no such set of $D+1$ columns can exist in $M$, and $s_{\text{min}}(M) \geq D+2$.
\end{proof}

Because $s_{\text{min}}(M) \leq d_{\text{min}}(\mathcal{C})$, the following corollary is immediate.

\begin{corollary}
    Let $M$ be a $\bar{D}$-separable matrix. Then the code $\mathcal{C}(M)$ has minimum distance $d_{min}(\mathcal{C}) \geq D+2$.
\end{corollary}

We conclude the section on general results by making a graph-theoretic connection between $D$-disjunct and $D$-separable matrices and the Tanner graphs of the parity-check codes obtained from them.

\begin{remark} 
    Let $G = (V,W;E)$ be a bipartite graph with vertex sets $V$ and $W$ and edge set $E$ and let $M$ be its biadjacency matrix where columns are indexed by $V$ and rows are indexed by $W$. Notice that if $v \in V$ corresponds to a column in $M$, then the set of neighbors of $v$, $\mathcal{N}(v)$, is exactly the set of nonzero rows in column $v$. Let $\{v_i\}_{i=1}^t \subseteq V$ be a collection of $t$ vertices in $V$. Because $\bigcup_{i=1}^t \mathcal{N}(v_i) = \mathcal{N}\left( \{v_i\}_{i=1}^t\right)$, the Boolean sum of the $t$ columns corresponding to $\{v_i\}_{i=1}^t$ is the same as $\mathcal{N}\left(\{v_i\}_{i=1}^t\right)$. From this observation and Definitions~\ref{def:disjunct} and \ref{def:separable}, it is clear that $M$ is
    \begin{itemize}
        \item $D$-disjunct if and only if given any set of $D$ vertices $\{v_i\}_{i=1}^D \subseteq V$, for any $v \in V \setminus \{v_i\}_{i=1}^D$, $\mathcal{N}(v) \not\subseteq \mathcal{N}\left( \{v_i\}_{i=1}^D \right)$
        \item $\bar{D}$-separable if and only if given any two sets of $t \leq D$ vertices $\{v_i\}_{i=1}^t \subseteq V$ and $\{w_i\}_{i=1}^t \subseteq V$ with $\{v_i\}_{i=1}^t \neq \{w_i\}_{i=1}^t$, $\mathcal{N}\left(\{v_i\}_{i=1}^t\right) \neq \mathcal{N}\left(\{w_i\}_{i=1}^t\right)$.
        \item $D$-separable if and only if give any two sets of $D$ vertices $\{v_i\}_{i=1}^D \subseteq V$ and $\{w_i\}_{i=1}^D \subseteq V$ with $\{v_i\}_{i=1}^D \neq \{w_i\}_{i=1}^D$, $\mathcal{N}\left(\{v_i\}_{i=1}^D\right) \neq \mathcal{N}\left(\{w_i\}_{i=1}^D\right)$.
    \end{itemize}
\end{remark}

\section{Specific Constructions of Disjunct Matrices} \label{sec:specific}

In this section, we will examine the parameters of parity-check codes from various constructions of $D$-disjunct matrices. We will look at constructions from three sources: Macula \cite{macula1996simple}, and Fu \& Hwang \cite{fu2006novel}, and Kautz \& Singleton \cite{kautz1964nonrandom}. Although Kautz \& Singleton have multiple constructions given in their paper, we consider using the matrices of superimposed codes built from $q$-ary codes as parity-check codes since that construction is a fundamental tool in the area of group testing (see \cite{barg2017group, du2000combinatorial, mazumdar2016nonadaptive} and references therein). The other constructions in the paper \cite{kautz1964nonrandom}  directly use binary error-correcting codes and codes based on block designs, which are well-studied elsewhere.

For each of these constructions, we show that the Tanner graphs corresponding to the disjunct matrices have girth $6$, which corresponds to a column overlap $a_{\max} = 1$. Hence, by Proposition~\ref{prop:tillich}, and the discussion that follows it, these matrices have the best possible error-correction capabilities after one round of bit-flipping under majority-logic decoding.

\subsection{Macula's Subset Construction} \label{sec:macula}

In \cite{macula1996simple}, Macula gives a construction of matrices that are $D$-disjunct using subset inclusion properties. When used as the parity-check matrix of error-correcting codes for certain cases of $D$ and $K$, the construction results in LDPC codes with girth 6.  We briefly review the construction from \cite{macula1996simple}. 

Let $\left({N\choose j}\right)$ denote the set of $j$-subsets of $[N]$. For $D<K<N$, define the binary matrix $M(D, K, N)$ to be a ${N\choose D} \times {N\choose K}$ matrix whose rows and columns (resp.) are indexed by the elements (sets) of $\left({N\choose D}\right)$ and $\left({N\choose K}\right)$. For a set $S\in \left({N\choose D}\right)$ and $U\in \left({N\choose K}\right)$, the matrix $M(D, K, N)$ has a 1 in the $(S,U)$th position if and only if $S\subseteq U$. The column weight of $M$ is ${K\choose D}$ and the row weight is ${N-D\choose N-K}$. 

Consider the specific case when $K=D+1$. Then the rows of $M$ correspond to $D$-subsets of $[N]$ and the columns correspond to $D+1$-subsets, the column weight is $D+1$ and the row weight is $N-D$. 

\begin{remark} \label{rem:half-n} For $D>\frac{N}{2}$, the matrix $M(D, D+1, N)$ would be a matrix with more rows than columns, since ${N\choose D}>{N\choose D+1}$ for those values. Because parity-check matrices generally have more columns than rows, we restrict ourselves to the case when $D \leq \frac{N}{2}$.
\end{remark}

We begin by considering the girth of the Tanner graphs of the codes from $M(D,D+1,N)$ matrices.

\begin{proposition} \label{prop:macula4cycles}
The Tanner graph of the matrix $M(D, D+1, N)$ is 4-cycle-free.  
\end{proposition}

\begin{proof}
    For the sake of contradiction, assume that the Tanner graph of $M(D, D+1, N)$ contains a cycle consisting of check nodes $S_1, S_2$ and variable nodes $U_1, U_2$, where $S_1\neq S_2$ and $U_1\neq U_2$. %We will show that $U_1=U_2$. 
    By the definition of the matrix $M$, this means that $S_i\subseteq U_j$, $i=1,2, j=1, 2$. In order for $S_1$ and $S_2$ to be distinct $D$-subsets of $[N]$, they must differ in at least one element. Thus $|S_1\cup S_2|\geq D+1$. Similarly, $U_1$ and $U_2$ must differ in at least one element. 

    Since $S_1\subseteq U_1$ and $S_2\subseteq U_1$, we have that $S_1\cup S_2\subseteq U_1$. Thus $|S_1\cup S_2|\leq |U_1|$. Therefore  
    \[ D+1\leq |S_1\cup S_2|\leq |U_1|=D+1,\]
which implies that $U_1=S_1\cup S_2$. By an analogous argument, $U_2=S_1\cup S_2$. But then $U_1= U_2$, a contradiction.\end{proof}

\begin{proposition} \label{prop:macula6cycles}
    The Tanner graph of the matrix $M(D, D+1, N)$ contains a 6-cycle. 
\end{proposition}

\begin{proof}
    Consider variable nodes indexed by the following $D+1$-sets, for example: 
    \begin{itemize}
        \item $U_1=\{1, 2, \ldots, D+1\}$
        \item $U_2=\{1, 2, \ldots, D, D+2\}$
        \item $U_3=\{2, 3, \ldots, D+2\}$
    \end{itemize}
    and check nodes indexed by the $D$-sets: 
    \begin{itemize}
        \item $S_1=\{1, 2, \ldots, D\}$
        \item $S_2=\{2, 3, \ldots, D, D+2\}$
        \item $S_3=\{2, 3, \ldots, D, D+1\}$.
    \end{itemize}
    Then the cycle $U_1, S_1, U_2, S_2, U_3, S_3, U_1$ is a 6-cycle in the Tanner graph. 
\end{proof}

The following corollary is immediate from Propositions~\ref{prop:macula4cycles} and \ref{prop:macula6cycles}.

\begin{corollary}
    The Tanner graph of the matrix $M(D,D+1,N)$ has girth $6$.
\end{corollary}

Due to the highly structured nature of $M(D,D+1,N)$, we can count the number of $6$-cycles in its corresponding Tanner graph.

\begin{proposition}
    The number of $6$-cycles in $M(D, D+1, N)$ is ${N\choose D+2}{D+2\choose 3}$. 
\end{proposition}

\begin{proof}
   Suppose $S_1$, $S_2$, and $S_3$ are $D$-subsets that are part of a $6$-cycle in the Tanner graph of $M(D, D+1, N)$. To have a $6$-cycle, we need $D+1$-sets $U_1$, $U_2$, and $U_3$ such that 
    \begin{multicols}{3}
        \begin{itemize}
            \item $S_1 \subseteq U_1$
            \item $S_1 \subseteq U_3$
            \item $S_2 \subseteq U_2$
            \item $S_2 \subseteq U_3$
            \item $S_3 \subseteq U_1$
            \item $S_3 \subseteq U_2$
        \end{itemize}
    \end{multicols}
    If $S_1 \subseteq U_3$ and $S_2 \subseteq U_3$, then $S_1 \cup S_2 \subseteq U_3$, so $|S_1 \cup S_2 | \leq D+1$. Because $S_1$ and $S_2$ are both $D$-sets and are not equal, this means that $S_1$ and $S_2$ differ in exactly one item. A similar argument holds for $S_1$ and $S_3$ as well as for $S_2$ and $S_3$. Hence the collection of sets $\{S_i\}_{i=1}^3$ must differ pairwise in exactly one item.
    
    To see that the number of $6$-cycles in the Tanner graph of $M(D, D+1, N)$ equals ${N\choose D+2}{D+2\choose 3}$, consider the number of sets $S_1, S_2, S_3, U_1, U_2, U_3$ described above. %in part (2) of the Proposition. 
    The number of sets of the form $S_1\cup S_2\cup S_3$ is  ${N\choose D+2}$. Three of the elements from that $(D+2)$-subset are chosen to be the elements that distinguish $S_1, S_2, S_3$.  There are ${D+2\choose 3}$ possibilities for these special elements. Their selection, along with the selection of the $D-1$ elements in $S_1\cap S_2\cap S_3$, determines the 6-cycle described above. Thus the total number of ways to construct such a $6$-cycle is 
    \[ {N\choose D+2}{D+2\choose 3}.\]

\end{proof}

\begin{remark}
    For $K\geq D+2$, the Tanner graph of $M(D, K, N)$ contains 4-cycles because the nodes indexed by sets 
    \begin{itemize}
        \item $S_1=\{1, 2, \ldots, D\}$, 
        \item $S_2=\{1, 2, \ldots, D-1, D+1\}$, 
        \item $U_1= \{1, 2, \ldots, D, D+1, D+2\}$,
        \item $U_2=\{1, 2, \ldots, D, D+1,D+3\}$,
        \end{itemize}
        form a 4-cycle $U_1, S_1, U_2, S_2, U_1$. 
\end{remark}

We next explore the rate and distance parameters of these codes.

\begin{proposition} 
    Let $\mathcal{C}$ be the code with parity check matrix $M(D,D+1,N)$. Assuming that $D \leq \frac{N}{2}$, the rate $R$ of $\mathcal{C}$ is $R \geq 1 - \frac{D+1}{N-D}$.
\end{proposition}

\begin{proof}
    By construction, $M(D,D+1,N)$ is an $\binom{N}{D} \times \binom{N}{D+1}$ matrix. Because this matrix could have redundant rows, the rate of the matrix is lower bounded by $1 - \frac{\binom{N}{D}}{\binom{N}{D+1}}$, which simplifies to $1 - \frac{D+1}{N-D}$.
\end{proof}

\begin{theorem} \label{thm:mindis}
Let $\mathcal{C}$ be the code with parity check matrix $M(D,D+1,N)$. The minimum distance of $\mathcal{C}$ is $d_{\text{min}}(\mathcal{C}) = D+2$.
\end{theorem}

\begin{proof}
    We will exhibit a set of $D+2$ linearly dependent columns of $M = M(D,D+1,N)$ to show that $d_{\text{min}}(\mathcal{C}) \leq D+2$.

    \begin{itemize}
        \item Let $\mathcal{S}^{(N-D+1)}$ be the collection of $D$-subsets of $[N]$ whose smallest element is $N-D+1$. Note that $|\mathcal{S}^{(N-D+1)}| = 1$. We will call this singular element $S^{(N-D+1)}_0$.
        \item Let $\mathcal{S}^{(N-D)}$ be the collection of $D$-subsets of $[N]$ whose smallest element is $N-D$. Note that $|\mathcal{S}^{(N-D)}| = D$. Let $S^{(N-D)}_i$ be the element of $\mathcal{S}^{(N-D)}$ that is missing element $i$.
        \item Let $\mathcal{S}^{(N-D-1)}$ be the collection of $D$-subsets of $[N]$ whose smallest element is $N-D-1$. Note that $|\mathcal{S}^{(N-D-1)}| = \binom{D+1}{2}$. Let $S^{(N-D)}_{i,j}$ be the element of $\mathcal{S}^{(N-D-1)}$ that is missing elements $i$ and $j$.
        \item Let $\mathcal{U}^{(N-D)}$ be the collection of $(D+1)$-subsets of $[N]$ whose smallest element is $N-D$. Note that $|\mathcal{U}^{(N-D)}| = 1$. We will call this singular element $U^{(N-D)}_0$.
        \item Let $\mathcal{U}^{(N-D-1)}$ be the collection of $(D+1)$-subsets of $[N]$ whose smallest element is $N-D-1$. Note that $|\mathcal{U}^{(N-D-1)}| = D+1$. Let $U^{(N-D-1)}_i$ be the element of $\mathcal{U}^{(N-D-1)}$ that is missing element $i$.
        \item Let $\mathcal{S} = \bigcup_{i=N-D-1}^{N-D+1} \mathcal{S}^{(i)}$.
        \item Let $\mathcal{U} = \bigcup_{i=N-D-1}^{N-D} \mathcal{U}^{(i)}$.
    \end{itemize}

    We will show that each set $S \in \mathcal{S}$ is a subset of exactly two sets in $\mathcal{U}$.

    \begin{itemize}
        \item Consider $S^{(N-D+1)}_0 \in \mathcal{S}^{(N-D+1)}$. Clearly, $S^{(N-D+1)}_0 \subsetneq U^{(N-D)}_0$. Additionally, $S^{(N-D+1)}_0 \subsetneq U^{(N-D-1)}_{N-D}$.
        \item Consider $S^{(N-D)}_i \in \mathcal{S}^{(N-D)}$. Clearly, $S^{(N-D)}_i \subsetneq U^{(N-D)}_0$. Additionally, $S^{(N-D)}_i \subsetneq U^{(N-D-1)}_i$.
        \item Consider $S^{(N-D-1)}_{i,j} \in \mathcal{S}^{(N-D-1)}$. Then $S^{(N-D-1)}_{i,j} \subsetneq U^{(N-D-1)}_i$ and $S^{(N-D-1)}_{i,j} \subsetneq U^{(N-D-1)}_j$.
    \end{itemize}

    Because $|\mathcal{S}| = 1 + D + \binom{D+1}{2} = 1 + D + \frac{D}{2}(D+1)$ and each of these sets was contained in two sets in $\mathcal{U}$, we have accounted for
    \begin{align*}
        2(1+D+\frac{D}{2}(D+1)) &= 2 + 2D + D(D+1) = D^2 + 3D + 2
    \end{align*}
    ones. Recall that the column weight of any column of $M$ is $D+1$. So there are $(D+2)(D+1) = D^2 + 3D + 2$ ones in these columns of $M$. Hence we have accounted for all nonzero entries in the columns corresponding to $\mathcal{U}$.

    Because each row indexed by $\mathcal{S}$ has exactly two ones in the columns indexed by $\mathcal{U}$ and this accounts for all ones in these columns, all columns indexed by $\mathcal{U}$ have row sums equivalent to $0$ mod $2$. This shows that the collection of $D+2$ columns adds up to $\mathbf{0}$ mod $2$, which gives us that $d_{\text{min}}(\mathcal{C}) \leq D+2$.
    
    By Proposition~\ref{prop:mindis}, $d_{\text{min}}(\mathcal{C}) \geq D+2$. So $d_{\text{min}}(\mathcal{C}) = D+2$.
\end{proof}

The following result on stopping distance is relatively immediate.

\begin{corollary}
    Let $\mathcal{C}$ be the code with parity check matrix $M = M(D,D+1,N)$. The minimum stopping distance of $M$ is $s_{\text{min}}(M) = D+2$.
\end{corollary}

\begin{proof}
    By Theorem~\ref{thm:mindis}, $d_{\text{min}}(\mathcal{C}) = D+2$. By Proposition~\ref{prop:stoppingdis}, $s_{\text{min}}(M) \geq D+2$. Because $s_{\text{min}}(M) \leq d_{\text{min}}(\mathcal{C})$, we have $s_{\text{min}}(M) = D+2$.
\end{proof}

We conclude this subsection with a discussion on the density of the matrices $M(D, D+1, N)$. The density of ones in $M(D, D+1, N)$ is $\frac{D+1}{{N\choose D}}$, since there are $D+1$ ones per column, and each column is length ${N\choose D}$. A code family with fewer than 10 ones per row can be classified as LDPC \cite{misoczki2013mdpc}, so suppose $N-D\leq 10$, so $N\leq D+10$. Coupled with the observation about $D\leq N/2$ in Remark~\ref{rem:half-n}, these matrices $M(D, D+1, N)$ are low-density  for values of $N\leq 20$. More generally, the codes from $M(D, D+1, N)$ with $N>20$ would be MDPC. 

Letting $D$ be fixed, as $N \to \infty$, the number of columns of $M$ is $n = \binom{N}{D+1} = O(N^{D+1})$, and so the row weight $N-D$ of $M$ is $O(n^{1/(D+1)})$. In order to be strictly considered MDPC, these codes would need $D = 1$. However, any $D$ value gives a density of at most $O(\sqrt{n})$, and so Proposition~\ref{prop:tillich} holds for all codes constructed from $M(D,D+1,N)$. In other words, all codes in this family can correct at least $\left\lfloor \frac{N-D}{2} \right\rfloor$ errors under a single round of majority-logic bit-flipping decoding.

\subsection{$t$-Packing Subset Construction} \label{sec:tpacking}

In \cite{fu2006novel}, Fu and Hwang propose a new construction heavily inspired by the one given in \cite{macula1996simple} that makes use of $t$-packings. A $t$-packing is an ordered pair $(V,\mathcal{B})$ where $|V| = v$ and $\mathcal{B}$ is a collection of $k$-subsets (blocks) of $V$ such that every $t$-subset of $V$ occurs in at most one block of $\mathcal{B}$. In the specific case where each $t$-subset occurs in exactly one block of $\mathcal{B}$, $(V,\mathcal{B})$ is a Steiner $t$-design and is denoted $S(t,k,v)$.

Let $P(t,k,v)$ denote a $t$-packing with $v$ elements and blocks of size $k$. In \cite{fu2006novel}, the authors construct a $D$-disjunct matrix using $P(t,k,v)$. Their construction can be summarized as follows: for any $r$ with $1 \leq r \leq t-1$, let $M_r$ be a $\binom{v}{r} \times n$ binary matrix where the columns correspond to an arbitrary set of $n$ blocks from $P(t,k,v)$, the rows correspond to all the $\binom{v}{r}$ $r$-subsets, and entry $(i,j)$ is $1$ if and only if the $r$-subset in row $i$ is a subset of the block corresponding to column $j$. Otherwise, entry $(i,j)$ is zero.\footnote{Notice that when you take $t=k$, $r=t-1$, and use every block, the matrix $M_r$ coincides with the matrix $M(t-1, t, v)$ from Subsection~\ref{sec:macula}. In the current subsection, we will consider the separate case when $k=t+1$. }

\begin{theorem} [\cite{fu2006novel}]
    $M_r$ is a $D$-disjunct matrix with $D = \left\lceil \binom{k}{r} / \binom{t-1}{r} \right\rceil - 1$ for $1 \leq r \leq t-1$. Further, $D$ is maximized when $r = t-1$.
\end{theorem}

We will limit our discussion to the case where $r = t-1$ and $k = t+1$, and so $D = \binom{t+1}{t-1} - 1 = \frac{t(t+1)}{2} - 1$. These choices of $r$ and $k$ not only maximize the possible disjunct value, but also assure that the Tanner graph corresponding to $M_r$ has girth $6$, which we will show in Proposition~\ref{prop:tpackingcycles}.

\begin{proposition} \label{prop-m-r-4-cycle}
    Let $M_{t-1}$ be a matrix created as described above and based on some $P(t,t+1,v)$. Then the Tanner graph corresponding to $M_{t-1}$ is free of $4$-cycles, and so has girth at least $6$.
\end{proposition}

\begin{proof}
    Let $S_1$ and $S_2$ (with $S_1 \neq S_2$) correspond to arbitrary rows of $M_{t-1}$. Note that $|S_1| = |S_2| = t-1$ and, because $S_1 \neq S_2$, $|S_1 \cup S_2| \geq t$.

    Let $U_1$ and $U_2$ (with $U_1 \neq U_2$) correspond to arbitrary columns of $M_{t-1}$. Because both $U_1$ and $U_2$ are blocks from a $t$-packing, they intersect in at most $t-1$ entries, so $|U_1 \cap U_2| \leq t-1$. 

    If the Tanner graph corresponding to $M_{t-1}$ has a $4$-cycle, then some $S_1$ and $S_2$ are both contained in some $U_1$ and $U_2$. Because $|S_1 \cup S_2| \geq t$, this would mean $|U_1 \cap U_2| \geq t$, a contradiction to $U_1$ and $U_2$ being blocks from a $t$-packing. Hence no $4$-cycles exist, and so the Tanner graph has girth at least $6$.
\end{proof}

\begin{remark}
    Consider an arbitrary $t$-packing $P(t,k,v) = (V,\mathcal{B})$. If $k = t$, then all elements of $\mathcal{B}$ are disjoint. In this case, any $r$-subset where $1 \leq r \leq t-1$ will be contained in at most one block in $\mathcal{B}$. While a Tanner graph corresponding to such a matrix would have no cycles, it would also have row weights at most $1$, which yields a code of minimum distance $0$ or $1$.
\end{remark}

In general, given an arbitrary $t$-packing, one cannot say more about the girth of the corresponding Tanner graph. However, if we restrict to Steiner $t$-designs, $S(t,t+1,v)$ and assume we take all blocks instead of an arbitrary subset of blocks, we can show that $M_{t-1}$ has girth $6$.

But first, we examine some additional properties of $M_{t-1}$ that we can say if it is based on an $S(t,t+1,v) = (V,\mathcal{B})$ and uses all blocks. We require some combinatorial design theory results, taken from \cite{colbourn2010crc}.

\begin{theorem} \cite{colbourn2010crc} \label{thm:colbourn}
    Let $(V,\mathcal{B})$ be an $S(t,k,v)$. Then:
    \begin{enumerate}
        \item $|\mathcal{B}| = \binom{v}{t} / \binom{k}{t}$.
        \item Each $s$-subset of $V$ is contained in $\binom{v-s}{t-s}/\binom{k-s}{t-s}$ blocks.
    \end{enumerate}
\end{theorem}

\begin{proposition} \label{prop:mrparams}
    Let $M_{t-1}$ be created as described, but taking all columns instead of a subset, and based on some $S(t,t+1,v)$. Then the following hold:
    \begin{enumerate}
        \item $M_{t-1}$ is a $\binom{v}{t-1} \times \binom{v}{t} / (t+1)$ matrix.
        \item $M_{t-1}$ has row weight $\frac{v-t+1}{2}$.
        \item $M_{t-1}$ has column weight $\binom{t+1}{t-1}$.
    \end{enumerate}
\end{proposition}

\begin{proof}
    From Theorem~\ref{thm:colbourn}, $|\mathcal{B}| = \binom{v}{t} / \binom{k}{t} = \binom{v}{t}/\binom{t+1}{t} = \binom{v}{t} / (t+1)$, and so $M_{t-1}$ has $\binom{v}{t}/(t+1)$ columns and is hence a $\binom{v}{t-1} \times \binom{v}{t}/(t+1)$ matrix, proving (1).

    Also from Theorem~\ref{thm:colbourn}, each $(t-1)$-subset of $V$ is contained in $\binom{v-(t-1)}{t-(t-1)}/\binom{(t+1)-(t-1)}{t-(t-1)} = \frac{v-t+1}{2}$ blocks, and so $M_{t-1}$ has row weight $\frac{v-t+1}{2}$, proving (2).

    Finally, because each column set is size $t+1$ and each row set is size $t-1$, each column set contains $\binom{t+1}{t-1}$ subsets of size $(t-1)$, and so has column weight $\binom{t+1}{t-1}$, proving (3).
\end{proof}

\begin{proposition} \label{prop:tpackingcycles}
    Let $M_{t-1}$ be created as described, but taking all columns instead of a subset, and based on some $S(t,t+1,v) = (V,\mathcal{B})$. Then:
    \begin{enumerate}
        \item the Tanner graph corresponding to $M_{t-1}$ has girth $6$.
        \item every triple of distinct $(t-1)$-subsets of $V$ that have exactly one item different (i.e. $|S_i \cap S_j \cap S_k| = t-2$) corresponds to a $6$-cycle in the Tanner graph corresponding to $M_{t-1}$.
        \item the Tanner graph corresponding to $M_{t-1}$ has exactly ${v\choose t+1}{t+1\choose 3}$ $6$-cycles.
    \end{enumerate}
\end{proposition}

\begin{proof}
    Let $S(t,t+1,v) = (V,\mathcal{B})$ be a Steiner $t$-design. Consider the following $(t-1)$ subsets of $V$ (corresponding to rows of $M_{t-1}$) and $(t+1)$-subsets of $V$ (corresponding to elements of $\mathcal{B}$):
    \begin{multicols}{2}
    \begin{itemize}
        \item $S_1 = \{x_1, \dotsc, x_{t-1}\}$
        \item $S_2 = \{x_1, \dotsc, x_{t-2}, x_t\}$
        \item $S_3 = \{x_1, \dotsc, x_{t-2}, x_{t+1}\}$
        \item $U_1 = \{x_1, \dotsc, x_{t-1}, x_{t+1}, y_1\}$
        \item $U_2 = \{x_1, \dotsc, x_{t-2}, x_t, x_{t+1}, y_2\}$
        \item $U_3 = \{x_1, \dotsc, x_{t-1}, x_t, y_3\}$
    \end{itemize}
    \end{multicols}
    \noindent where the $x_i$ are arbitrary elements of $V$ and the $y_i$ are determined by the specific structure of $(V,\mathcal{B})$. We know that each of the $U_i$ exist, because $(V,\mathcal{B})$ is such that each $t$-subset of $V$ is contained in exactly one block. So $U_1$ is the unique block of $(V,\mathcal{B})$ containing elements $x_1, \dotsc, x_{t-1}, x_{t+1}$, $U_2$ is the unique block containing elements $x_1, \dotsc, x_{t-2}, x_t, x_{t+1}$, and $U_3$ is the unique block containing elements $x_1, \dotsc, x_{t-1}, x_t$. Notice that:
    \begin{multicols}{3}
    \begin{itemize}
        \item $S_1 \subseteq U_1$
        \item $S_2 \subseteq U_2$
        \item $S_3 \subseteq U_1$
        \item $S_1 \subseteq U_3$
        \item $S_2 \subseteq U_3$
        \item $S_3 \subseteq U_2$.
    \end{itemize}
    \end{multicols}
    Consequently, $S_1, U_1, S_3, U_2, S_2, U_3, S_1$ is a $6$-cycle in the Tanner graph corresponding to $M_r$. Proposition~\ref{prop-m-r-4-cycle} shows that the girth is at least 6. So this Tanner graph has girth $6$, proving (1). 

    We will now show that every $6$-cycle in the Tanner graph corresponding to $M_{t-1}$ must come from a triple of distinct $(t-1)$-subsets of $V$ that have exactly one item different. A $6$-cycle in this Tanner graph must involve exactly three $(t-1)$-subsets of $V$. Suppose $S_1$, $S_2$, and $S_3$ are $(t-1)$-subsets that are part of a $6$-cycle. To have a $6$-cycle, we need some blocks $U_1$, $U_2$, and $U_3$ such that 
    \begin{multicols}{3}
        \begin{itemize}
            \item $S_1 \subseteq U_1$
            \item $S_1 \subseteq U_3$
            \item $S_2 \subseteq U_2$
            \item $S_2 \subseteq U_3$
            \item $S_3 \subseteq U_1$
            \item $S_3 \subseteq U_2$
        \end{itemize}
    \end{multicols}
    If $S_1 \subseteq U_3$ and $S_2 \subseteq U_3$, then $S_1 \cup S_2 \subseteq U_3$, so $|S_1 \cup S_2 | \leq t+1$. Because $S_1$ and $S_2$ are both $(t-1)$-sets and are not equal, this means that $S_1$ and $S_2$ differ in exactly one item. A similar argument holds for $S_1$ and $S_3$ as well as for $S_2$ and $S_3$. Hence the collection $\{S_i\}_{i=1}^3$ must differ in exactly one item. This proves (2).

    To see that the number of $6$-cycles in the Tanner graph of $M_{t-1}$ equals ${v\choose t+1}{t+1\choose 3}$, consider the number of sets described in part (2) of the Proposition. The number of sets of the form $S_1\cup S_2\cup S_3$ is  ${v\choose t+1}$. Three of the elements from that $(t+1)$-subset are chosen to be the elements that distinguish $S_1, S_2, S_3$.  There are ${t+1\choose 3}$ possibilities for these special elements. Their selection, along with the selection of the $t-2$ elements in $S_1\cap S_2\cap S_3$, determines the 6-cycle described in part (2). Thus the total number of ways to construct such a $6$-cycle is 
    \[ {v\choose t+1}{t+1\choose 3}.\]
\end{proof}

\begin{remark}
    If $1 \leq r \leq t-2$, the Tanner graph corresponding to $M_r$ using all columns from an $S(t,k,v)$ will have $4$-cycles for any choice of $k \geq t+1$. We will show that this is true for $r=t-2$ and then argue that this is sufficient for all $1 \leq r \leq t-2$. Assume that $V = [v]$ and consider the following $(t-2)$-subsets of $V$: 
    \begin{multicols}{2} 
        \begin{itemize}
            \item $S_1 = \{1, 2, \dotsc, t-2\}$
            \item $S_2 = \{1, 2, \dotsc, t-3, t-1\}$
        \end{itemize}
    \end{multicols}
    \noindent
    and the following blocks of $S(t,k,v)$:
    \begin{multicols}{2}
        \begin{itemize}
            \item $B_1 = \{1, 2, \dotsc, t, x_{t+1}, \dotsc, x_k\}$
            \item $B_2 = \{1, 2, \dotsc, t-1, x_{t+1}, y_1, \dotsc, y_k\}$
        \end{itemize}
    \end{multicols}
    Because $S_1 \subsetneq B_1$, $S_1 \subsetneq B_2$, $S_2 \subsetneq B_1$, and $S_2 \subsetneq B_2$, $S_1, B_1, S_2, B_2, S_1$ is a $4$-cycle in the Tanner graph corresponding to $M_r$. 

    Notice that if $1 \leq r < t-2$, there exist $r$-subsets of $S_1$ and $S_2$ fulfilling the same conditions as above. Hence $4$-cycles exist in Tanner graphs corresponding to $M_r$ using all columns from an $S(t,k,v)$ for all $1 \leq r \leq t-2$. 
\end{remark}

We will now examine what we can say about the rate and distance parameters of the class of parity-check codes taken from an $M_{t-1}$ based on an $S(t,t+1,v)$.

\begin{proposition}
    Let $\mathcal{C}$ be the code with parity check matrix $M_{t-1}$ constructed using all blocks of $S(t,t+1,v)$. The rate $R$ of $\mathcal{C}$ is bounded by $R \geq 1 - \frac{t(t+1)}{v-t+1}$.
\end{proposition}

\begin{proof}
    According to Proposition~\ref{prop:mrparams}, $M_{t-1}$ is an $\binom{v}{t-1} \times \binom{v}{t}/(t+1)$ matrix. Because this matrix could have redundant rows, the rate of the code corresponding to $M_{t-1}$ is lower bounded by $1 - \frac{\binom{v}{t-1}}{\binom{v}{t}/(t+1)}$, which simplifies to $1 - \frac{t(t+1)}{v-t+1}$.
\end{proof}

By Proposition~\ref{prop:mindis}, $\mathcal{C}(M_{t-1})$ based on $S(t,t+1,v)$ has minimum distance $d_{\text{min}}(\mathcal{C}(M_{t-1})) \geq D + 2 = \frac{t(t+1)}{2}+1$. In general, not all designs will meet this bound, as noted in Example~\ref{exa:tpacking}.

\begin{example} \label{exa:tpacking}
    In the case of an $S(3,4,v)$, the lower bound on the minimum distance of $\mathcal{C}(M_2)$ is $\frac{3(3+1)}{2}+1 = 7$.
    \begin{itemize}
        \item If we build $M_2$ from the unique $S(3,4,10)$, $\mathcal{C}(M_2)$ is a $[30,2,11]$ code. In particular, this code does not meet the lower bound on the minimum possible distance of this construction.
        \item If we build $M_2$ from one\footnote{The $S(3,4,14)$ programmed natively into SageMath.} of the four nonisomorphic $S(3,4,14)$'s, $\mathcal{C}(M_2)$ is a $[91,7,7]$ code, showing that the minimum distance bound is tight in at least some situations.
    \end{itemize}
\end{example}

Note that given some $S(t,t+1,v) = (V,\mathcal{B})$, a codeword in this case corresponds to a subset of blocks $\mathcal{S} \subseteq \mathcal{B}$ such that each $(t-1)$-subset of the blocks in $\mathcal{S}$ is contained in an even number of the blocks in $\mathcal{S}$.

If we assume that such a configuration comes from the case where each subset of the blocks in $\mathcal{S}$ is contained in exactly $2$ blocks of $\mathcal{S}$, we can count that we need $1$ initial line $\ell_0$, and an additional $\binom{t+1}{2}$ lines (one for each $(t-1)$-subset of $\ell_0$). This is exactly the lower bound on the minimum distance. If any $(t-1)$-subset was contained in more than $2$ blocks, this would require more than $\binom{t+1}{2} + 1$ lines, and so the minimum size configuration that leads to a codeword must come from the case where each subset of the blocks in $\mathcal{S}$ is contained in \textit{exactly} two blocks of $\mathcal{S}$. This discussion is summarized in the following theorem.

\begin{theorem}
    Construct $M_{t-1}$ as described using all blocks of an $S(t,t+1,v) = (V,\mathcal{B})$. Then $d_{\text{min}}(\mathcal{C}(M_{t-1})) = \binom{t+1}{2} + 1$ if and only if there is some line configuration $\mathcal{S}$ with $\binom{t+1}{2} + 1$ lines so that $\mathcal{S} = \{B_i\}_{i=1}^{\binom{t+1}{2}+1} \subseteq \mathcal{B}$ such that for all $i$ and $j$ with $i \neq j$, $|B_i \cap B_j| = t-1$.
\end{theorem}

When $t = 2$, this configuration is exactly the well-studied Pasch configuration. In some sense, this collection of configurations could be considered generalizations of the Pasch configuration. Below are these configurations for $t=2$ and $t=3$.

\begin{table}[h]
    \centering
\begin{tabular}{|c|l|}
\hline
    $t$ & Blocks \\
    \hline
    $2$ & $012, 034, 145, 235$ \\
    $3$ & $0123, 0156, 0246, 0345, 1245, 1346, 2356$ \\
    \hline    
\end{tabular}
    \caption{\label{tab:blocks} Generalized Pasch Configurations in an $S(t,t+1,v)$ for $t=2$ and $t=3$.}
\end{table}

Infinite families of $S(t,t+1,v)$ exist for $t=2$ and $t=3$. A $S(2,3,v)$ exists if and only if $v \equiv 1 \pmod 6$ or $v \equiv 3 \pmod 6$. A $S(3,4,v)$ exists if and only if $v \equiv 2 \pmod 6$ or $v \equiv 4 \pmod 6$. Only finitely many $S(t,t+1,v)$ are known for $t \geq 4$, and none are known for $t \geq 6$ \cite{colbourn2010crc}. Note that when $t = 2$, the row sets are reduced to $1$-subsets of $v$, and so constructing $M_{t-1}$ from an $S(2,3,v)$ is equivalent to using the more well-studied incidence matrix. For information on parity-check codes that use incidence matrices of balanced incomplete block designs as their parity-check matrices, a good reference is \cite{johnson2003low}.

We conclude this subsection with a discussion on the density of these matrices. Let $M_r$ be constructed from a $t$-packing $P(t,t+1,v)$ with $r = t-1$. Each column has weight $\binom{t+1}{t-1}$ and $M_{t-1}$ has $\binom{v}{t-1}$ rows. So the density of $M_{t-1}$ is $\binom{t+1}{t-1}/\binom{v}{t-1} = \frac{(t+1)! (v-t+1)!}{2 v!}$. If we assume that the row weight of $M_{t-1}$ must be less than $10$ to be considered LDPC \cite{misoczki2013mdpc}, we need $\frac{v-t+1}{2} \leq 10$, which simplifies to the condition $v-t \leq 19$.

Letting $t$ be fixed, as $v \to \infty$, the number of columns of $M_{t-1}$ is $n = \binom{v}{t}/(t+1) = O(v^{t-1})$, and so the row weight $\frac{v-t+1}{2}$ is $O(n^{1/(t-1)})$. In order to be strictly considered MDPC, these codes need $t = 3$. When $t = 1$ or $t = 2$, these codes are too dense to be MDPC. However, any $t \geq 3$ gives a row weight of at most $O(\sqrt{n})$, and so Proposition~\ref{prop:tillich} holds for all codes constructed from $M_{t-1}$ with $t \geq 3$. In other words, all such codes with $t \geq 3$ can correct at least $\left\lfloor \frac{v-t+1}{4} \right\rfloor$ errors under a single round of majority-logic bit-flipping decoding.

\subsection{Disjunct Matrices from $q$-ary Codes} \label{sec:qary}

\label{sec:q-ary}

Kautz and Singleton \cite{kautz1964nonrandom} present a method of constructing binary superimposed codes from $q$-ary codes. This construction has since been used extensively for group testing \cite{barg2017group, hong2022group, inan2020strongly, mazumdar2012almost, mazumdar2016nonadaptive}, and was shown to be optimal for probabilistic group testing \cite{inan2019optimality}.
For clarity of notation, we will use $\mathcal{D}$ to denote the starting $[n,k,d]_q$ $q$-ary code.  The Kautz-Singleton construction is essentially a concatenation, where the resulting binary matrix has as its columns binary-expanded versions of each codeword in $\mathcal{D}$. Another way to interpret this construction in the context of the current work is as a method for obtaining a binary parity-check matrix from a $q$-ary code. 

The main example of the Kautz-Singleton construction uses a Reed-Solomon codes as the component code, $\mathcal{D}$. When the component code has dimension 2, the resulting binary matrix, viewed as a parity-check matrix, is $4$-cycle free and   moderate-density.

We now describe the main construction. Let $\mathcal{D}$ be an $[n, k, d]$ code over $\mathbb{F}_q$. Order the elements of $\mathbb{F}_q$: $\alpha_1, \alpha_2, \ldots, \alpha_q$. For $i=1, \ldots, q$, let $\textbf{e}_i$ denote the column vector with a 1 in position $i$ and zeroes in every other position. We form a binary matrix $M(\mathcal{D})$ by listing all codewords in $\mathcal{D}$ as column vectors, then replacing each symbol $\alpha_i$ with the column vector $\mathbf{e_i}$, for each element in the matrix and each $i=1, \ldots, q$.  

\begin{example} 
We start with a $[3, 2, 2]$ Reed-Solomon code over $\mathbb{F}_4=\{0, 1, \alpha, \alpha^2\}$ with generator matrix 
\[ G = \begin{pmatrix}
    1 & 1 & 1 \\ 
    0 & 1 & \alpha 
\end{pmatrix}\]
This code contains 16 codewords. Using the ordering $\alpha_1=0, \alpha_2=1, \alpha_3=\alpha, \alpha_4=\alpha^2$, we obtain the matrix $M(\mathcal{D})$. The codewords $\mathbf{c}_i$ correspond to columns in $M(\mathcal{D})$ as follows. Order the 16 polynomials $a+bx$, where $a,b\in \mathbb{F}_4$ as 
\[ 0, 1, \alpha, \alpha^2, x, x+1, x+\alpha, x+\alpha^2, \alpha x, \ldots, \alpha^2x+\alpha, \alpha^2x+\alpha^2.\]

For example, column 1 of the matrix is obtained from the evaluation of the polynomial $p_1(x)=0$ at the field elements $0, 1, \alpha$. Each codeword is then expanded to a binary string of length 12 in the matrix $M(\mathcal{D})$, via the process: $0$ is mapped to the binary column vector $\mathbf{e}_1$, $1$ to $\mathbf{e}_2$, $\alpha$ to $\mathbf{e}_3$, and $\alpha^2$ to $\mathbf{e}_4$. 
\end{example}

\begin{example} \label{ex:RS}
\setcounter{MaxMatrixCols}{20}

The matrix $M$ is $2$-disjunct.
\[ M=
\begin{pmatrix}                        
 1 & 0 & 0     & 0      & 1   & 0     & 0       & 0         & 1       & 0         & 0           & 0             & 1         & 0           & 0             & 0               \\ 
 0 & 1 & 0     & 0      & 0   & 1     & 0       & 0         & 0       & 1         & 0           & 0             & 0         & 1           & 0             & 0               \\ 
 0 & 0 & 1     & 0      & 0   & 0     & 1       & 0         & 0       & 0         & 1           & 0             & 0         & 0           & 1             & 0               \\ 
 0 & 0 & 0     & 1      & 0   & 0     & 0       & 1         & 0       & 0         & 0           & 1             & 0         & 0           & 0             & 1               \\  
 1 & 0 & 0     & 0      & 0   & 1     & 0       & 0         & 0       & 0         & 1           & 0             & 0         & 0           & 0             & 1               \\ 
 0 & 1 & 0     & 0      & 1   & 0     & 0       & 0         & 0       & 0         & 0           & 1             & 0         & 0           & 1             & 0               \\ 
 0 & 0 & 1     & 0      & 0   & 0     & 0       & 1         & 1       & 0         & 0           & 0             & 0         & 1           & 0             & 0               \\ 
 0 & 0 & 0     & 1      & 0   & 0     & 1       & 0         & 0       & 1         & 0           & 0             & 1         & 0           & 0             & 0               \\ 
 1 & 0 & 0     & 0      & 0   & 0     & 1       & 0         & 0       & 0         & 0           & 1             & 0         & 1           & 0             & 0               \\ 
 0 & 1 & 0     & 0      & 0   & 0     & 0       & 1         & 0       & 0         & 1           & 0             & 1         & 0           & 0             & 0               \\ 
 0 & 0 & 1     & 0      & 1   & 0     & 0       & 0         & 0       & 1         & 0           & 0             & 0         & 0           & 0             & 1               \\ 
 0 & 0 & 0     & 1      & 0   & 1     & 0       & 0         & 1       & 0         & 0           & 0             & 0         & 0           & 1             & 0               \\ 
    \end{pmatrix}
\]

When viewing $M$ as the parity-check matrix of a binary code $\mathcal{C}(M)$, the code has parameters $[16, 4, 8]$. Notice that by Proposition~\ref{prop:mindis} $d_{\min}(\mathcal{C}(M))\geq 4$, and in this example that bound is not met. 

\end{example}

Recall from Lemma~\ref{lemma-disjunct}, the disjunct value of a binary matrix satisfies $D=\left\lfloor \frac{w_{\min}-1}{a_{\max}}\right\rfloor$. 

For each $q$-ary code of length $n$, the matrix $M(\mathcal{D})$ constructed above has $n$ ones per column. Thus $w_{\min}=n$. Moreover, for any code of minimum distance $d$, there exist a pair of codewords that are exactly distance $d$ apart---these codewords share an overlap of symbols in $n-d$ positions. Therefore for any $[n,k,d]$ $q$-ary code, the matrix $M(\mathcal{D})$ is $D=\lfloor\frac{n-1}{n-d}\rfloor$-disjunct. The denominator is minimized and $D$ is maximized when the starting code is MDS, i.e., when $n-d=k-1$. 

We can detect $4$-cycles in the matrix $M(\mathcal{D})$ by using the parameter $a_{\max}$ defined in Lemma~\ref{lemma-disjunct}. Notice that the Tanner graph of the matrix $M(\mathcal{D})$ contains a 4-cycle if and only if $a_{\max}\geq 2$. Thus, to obtain $M(\mathcal{D})$ without $4$-cycles, we must start with a code where $n-d=a_{\max}=1$. For any MDS code being used in this context, this is equivalent to the dimension of the code equaling 2. Moreover, when $a_{\max}=0$ the codewords share no common entries, so the minimum distance of the code $\mathcal{D}$ is $d=n$, which means that $k=1$.  In this case, the matrix $M(\mathcal{D})$ is an $nq\times q$ matrix with $n$ block $q\times q$ permutation matrices, and thus the code $\mathcal{C}(M(\mathcal{D}))$ is the trivial code.

Relating this to the Kautz-Singleton construction for Reed-Solomon codes in particular, the dimension of the code is the bound on the polynomial degree for the evaluation code. Codewords correspond to polynomials of degree up to $k-1$, thus when $k>2$ that includes quadratic polynomials. Since we are guaranteed to encounter pairs of distinct quadratics that intersect in two of the evaluation points, we will obtain a $4$-cycle in the Tanner graph in all cases when $2\leq n\leq q$ and $k>2$. When $k\leq 2$, the polynomials have degree at most 1, and thus distinct polynomials can intersect at either one point or zero points.
Therefore we have the following result. 

\begin{proposition}
    The Tanner graph corresponding to $M(\mathcal{D})$ where $\mathcal{D}$ is an $[n, k\leq 2, d]_q$ Reed-Solomon code and $M(\mathcal{D})$ is the binary matrix obtained as above, is free of 4-cycles, and so has girth at least 6. 
\end{proposition}

Use the notation $\mathcal{C}$ for a binary code with parity-check matrix $H=M(\mathcal{D})$, where $\mathcal{D}$ is a $q$-ary $[n, k, d]$ code used to construct $M(\mathcal{D})$ as in the Kautz-Singleton construction. Then $\mathcal{C}$ has parameters $[q^k, \geq q^k-nq, d_{\min}(\mathcal{C})]$. 

\begin{proposition} \label{prop:ks-rate}
    The rate of the binary code $\mathcal{C}$ with parity-check matrix $M(\mathcal{D})$ from the Kautz-Singleton construction is 
    \[ R\geq 1-\frac{n}{q^{k-1}}.\]
\end{proposition}

\begin{proof}
    Since $M(\mathcal{D})$ has $nq$ rows and $q^k$ columns, we have 
    \[ R\geq \frac{q^k-nq}{q^k}=\frac{q^{k-1}-n}{q^{k-1}}=1-\frac{n}{q^{k-1}}. \]
\end{proof}

\begin{corollary}
    By Proposition~\ref{prop:mindis}, $d_{\min}(\mathcal{C})\geq D+2=\left\lfloor\frac{n-1}{n-d}\right\rfloor+2$. 
\end{corollary}

In practice, the value of $d_{\min}(\mathcal{C})$ is greater than the general bound for $D$-disjunct matrices;  see  Example~\ref{ex:RS}.  

\begin{proposition} The value $d_{\min}(\mathcal{C}(M(\mathcal{D})))$ equals the size of a minimal set of columns of $M(\mathcal{D})$ such that the binary sum of the columns is the zero vector. 
    When $q=2^r$, a necessary but not sufficient condition for a set of columns $\mathcal{S}$ of the matrix $M(\mathcal{D})$ to sum to zero is that the sum of the corresponding codewords in $\mathcal{D}$ must be the zero codeword. 
\end{proposition}

\begin{proof} The first part of the proposition translates to the context of the matrix $M(\mathcal{D})$ the well-known relationship between a minimal linear dependence among the columns of a parity-check matrix $H$  to the minimum distance of the code. 

    The condition is necessary because if the sum of the columns of $M(\mathcal{D})$ is the zero vector, then each codeword symbol appears an even number of times in each position, meaning that the corresponding original codewords in $\mathcal{D}$ must sum to zero,  since the field has characteristic 2. 

    The condition is not sufficient because a collection of codewords in $\mathcal{D}$ may sum to zero in $\mathbb{F}_q^n$ while the sum of the binary columns in $M(\mathcal{D})$ is not zero. 
    
    For example, consider the first four columns of $M$ in Example~\ref{ex:RS}, which originate with the constant polynomials $0, 1, \alpha, \alpha^2$, and thus yield codewords $(0, 0, 0)^T$, $(1, 1, 1)^T$, $(\alpha, \alpha, \alpha)^T$, $(\alpha^2, \alpha^2, \alpha^2)^T$. The sum of these codewords in $\mathcal{D}$ is the zero codeword, but the sum of the corresponding binary columns in $M(\mathcal{D})$ is not zero. 
\end{proof}

\begin{remark}
    It is an open problem to find $d_{\min}(\mathcal{C}(M(\mathcal{D})))$. Finding $d_{\min}$
    is equivalent to finding a smallest subset $S$ of codewords in $\mathcal{C}$ with the property that in each position $i=1, \ldots, n$, and for each $\alpha \in \mathbb{F}_q$, the number of codewords in $S$ with $c_i=\alpha$ is even. 
\end{remark}

We conclude the subsection with a discussion of the density of these matrices. Regarding the density of ones in the matrix $M(\mathcal{D})$: there are $n$ ones per column, and $q^k$ columns in $M(\mathcal{D})$, so the density of ones in the matrix is $\frac{nq^k}{nq(q^k)}=\frac{1}{q}$. 
The number of ones in row $i$ is either $q^k$ if every codeword has  position $i$ equal to 0, or $q^{k-1}$ otherwise. When the starting code $\mathcal{D}$ is Reed-Solomon or otherwise MDS, there are no positions fixed at zero for all codewords, so in general we consider the case where the row weight is $q^{k-1}$ for each row. 

Codes with the family $M(\mathcal{D})$ as  parity-check matrices would be mostly moderate-density or high-density parity-check  codes. For fixed $q$, as $k\to \infty$, the row weight of $M(\mathcal{D})$ is $O(N)$. On the other hand, letting $k$ be fixed and allowing $q\to \infty$ (and thus $N\to \infty$), the row weight of $M(\mathcal{D})$ is  $O(\sqrt{N})$ when $k=2$.  Thus Proposition~\ref{prop:tillich} holds for codes constructed from a linear code $\mathcal{D}$ with parameters $[n, 2, d]_q$.  Further, if $\mathcal{D}$ is MDS, then $a_{\max}=k-1=1$. All codes with parity-check matrix $M(\mathcal{D})$ for such $\mathcal{D}$ can correct at least $\lfloor \frac{n}{2} \rfloor$ errors under a single round of majority-logic bit-flipping decoding, the optimal value for a code with parity-check matrix column weight $\lfloor \frac{n}{2} \rfloor$. The value of $n$ can be chosen to obtain a given rate for the code with parity-check matrix $M(\mathcal{D})$. For example,  a rate  $R\geq \frac{1}{2}$ would be guaranteed for $n=\frac{q}{2}$. In that case the number of correctable errors is $\frac{1}{4}$ the row weight of the matrix $M(\mathcal{D})$.

\section{Conclusions} \label{sec:conclusions}

There are two primary contributions of this work.

First, we introduced the idea of using disjunct and separable matrices as the parity-check matrices for parity-check codes. To this end, we provided bounds on the minimum distance and minimum stopping set size given an arbitrary $D$-disjunct or $D$-separable matrix as an initial avenue of research into this potential new class of constructions of parity-check codes. We also connected $D$-disjunct and $D$-separable properties of matrices to the Tanner graph representations of their corresponding parity-check codes.

Second, we provided three methods to algebraically construct codes from disjunct matrices that have optimal decoder performance under a single iteration of a bit-flipping majority-logic decoder. We indicated which parameters are required for these codes to be MDPC, a class of parity-check codes of particular interest for use in post-quantum cryptography.

For each of these three code families from disjunct matrices, we provided insight into the minimum distance, girth, stopping set size, and densities of their corresponding parity-check codes. Because of their highly structured nature, we were able to say the most about parity-check codes from Macula's subset construction in Section~\ref{sec:macula}. We were able to determine the exact minimum distance, stopping distance, and girth of these codes, and found that the stopping distance of these codes is the same as the minimum distance. Our exploration of the $t$-packing subset construction in Section~\ref{sec:tpacking} was not able to find an exact minimum distance, but the construction is flexible enough that some parameters could yield good codes. When considering parity-check codes from disjunct matrices from $q$-ary codes in Section~\ref{sec:qary}, we found that any parity-check code $\mathcal{C}(M(\mathcal{D}))$ built from an initial Reed-Solomon code $\mathcal{D}$ with dimension $k=2$ has girth at least 6 and has  optimal decoder performance under a single iteration of a bit-flipping majority-logic decoder.

This work opens up many new lines of questioning. There are additional classes of $D$-disjunct and $D$-separable matrices whose parameters for use in parity-check codes could be considered. While we considered stopping sets of the code families we looked at, we did not consider absorbing sets, which characterize the majority of decoder failures over the binary symmetric channel (BSC) \cite{dolecek2010absorbing}. For more information on classifying absorbing sets, see e.g. \cite{beemer2019absorbing, mcmillon2023extremal}. Finally, it would be interesting to examine how the structure of disjunct and separable matrices would affect iterative decoding algorithms and, in particular, if an efficient decoding algorithm specific to these matrices could be developed. As noted, the disjunct matrices in Section~\ref{sec:specific} with girth 6 have optimal decoder performance under a single iteration of a bit-flipping majority-logic decoder. It is an open question whether disjunct matrices yield optimal binary codes with respect to rate or minimum distance.  

In Section~\ref{sec:tpacking}, we noted that the configurations of $S(t,t+1,v)$'s that lead to minimum distance codewords in the given construction can be viewed as generalizations of the well-studied Pasch configuration in $S(2,3,v)$'s. The study of this generalization is an open avenue of research in design theory.

\section{Acknowledgements}

We would like to thank Tony Macula for suggesting that we look at his 1996 paper \cite{macula1996simple} for a class of matrices whose corresponding Tanner graphs have no $4$-cycles.

\bibliography{disjunct}
\bibliographystyle{abbrv} 

\end{document}